\documentclass[%
 aip,
 cha,%
 amsmath,amssymb,
reprint,%
]{revtex4-1}

\usepackage{graphicx}
\usepackage{dcolumn}
\usepackage{bm}
\usepackage{hyperref}
\usepackage{amsthm}
\usepackage{amsfonts}
\newtheorem{theorem}{Theorem}
\newtheorem{repeattheorem}{Theorem}
\newtheorem{proposition}{Proposition}
\newtheorem{assumption}{Assumption}
\newtheorem{definition}{Definition}
\newtheorem{lemma}{Lemma}
\newtheorem{corollary}{Corollary}

\begin{document}

\preprint{}

\title[Prediction uncertainty and optimal experimental design]{Prediction uncertainty and optimal experimental design\\ for learning dynamical systems}

\author{Benjamin Letham}
 \affiliation{Operations Research Center, Massachusetts Institute of Technology, Cambridge, MA 02139, USA}
 \email{bletham@alum.mit.edu}
 \altaffiliation[Present address: ]{Facebook, Menlo Park, CA 94025, USA}
\author{Portia A. Letham}
 \affiliation{Department of Chemical Engineering, Arizona State University, Tempe, AZ 85281, USA}
 \email{pletham@asu.edu}
\author{Cynthia Rudin}
 \affiliation{Department of Computer Science and Department of Electrical and Computer Engineering, Duke University, Durham, NC 27708, USA}
 \email{cynthia@cs.duke.edu}
\author{Edward P. Browne}
 \affiliation{Koch Institute for Integrative Cancer Research, Massachusetts Institute of Technology, Cambridge, MA 02139, USA}
 \email{ebrowne@broadinstitute.org}

\date{\today}

\begin{abstract}
Dynamical systems are frequently used to model biological systems. When these models are fit to data it is necessary to ascertain the uncertainty in the model fit. Here we present prediction deviation, a metric of uncertainty that determines the extent to which observed data have constrained the model's predictions. This is accomplished by solving an optimization problem that searches for a pair of models that each provide a good fit for the observed data, yet have maximally different predictions. We develop a method for estimating \textit{a priori} the impact that additional experiments would have on the prediction deviation, allowing the experimenter to design a set of experiments that would most reduce uncertainty. We use prediction deviation to assess uncertainty in a model of interferon-alpha inhibition of HIV infection, and to select a sequence of experiments that reduces this uncertainty. Finally we prove a theoretical result which shows that prediction deviation provides bounds on the trajectories of the underlying true model. These results show that prediction deviation is a meaningful metric of uncertainty that can be used for optimal experimental design.
\end{abstract}

\maketitle

\begin{quotation}
Nonlinear dynamical systems are used throughout systems biology to describe the dynamics of biomolecular interactions. These models typically have a number of unknown parameters, such as infection rates and decay rates, which are estimated by fitting the model to measurements from the physical system. Two important questions then arise: What is the uncertainty in the model predictions, and how can that uncertainty be reduced? We describe here a new approach for measuring uncertainty in model predictions, by searching for a pair of model parameters that both provide a good fit for the observed data, but make maximally different predictions. We further show how to estimate the impact on the uncertainty of a candidate experiment that has not yet been done, allowing the experimenter to determine beforehand if an experiment will be valuable. We use prediction deviation to analyze a model of HIV infection which can only be partially observed. With prediction deviation, and with appropriately selected experiments, we are able to provide bounds on the behavior of the unobserved quantities and gain insights into inhibition that are otherwise unavailable.
\end{quotation}

\section{Introduction}
Systems of nonlinear differential equations are used throughout biology to model the behavior of complex, dynamical systems. These models have proven particularly useful in systems biology for describing networks of biomolecular interactions \citep{Brackley10}. Often the utility of the model depends on being able to estimate a set of unknown parameters, which is typically done by collecting data from the physical system and finding the best-fit parameters. When inferring a dynamical system from data, there are two important questions that arise:
\begin{enumerate}
\item \textit{Uncertainty quantification}: Is the model sufficiently constrained by the data?
\item \textit{Optimal experimental design}: If not, what additional experiments would most reduce the remaining uncertainty?
\end{enumerate}
Uncertainty is often measured by constructing a confidence interval for each parameter estimate. We propose a different approach to the problem of uncertainty quantification, and then show that this approach leads naturally to an optimal experimental design strategy. Our fundamental hypothesis is that the purpose of fitting a model is to be able to use it to make predictions. In many situations the parameter values \textit{per se} are not of interest, rather the goal is to gain insight into the system's behavior. In these situations, the purpose of assessing model uncertainty is to determine if the model's predictions can be trusted.

This paper begins by developing \textit{prediction deviation}, a new measure of uncertainty on predicted behaviors. We then use prediction deviation to measure uncertainty in a partially observed model of HIV infection, where we found that after one experiment there remained substantial uncertainty in the behavior of the unobserved component. We then show that prediction deviation leads naturally to a way to measure \textit{experiment impact}, which is a maximum uncertainty on predicted behaviors if an additional experiment were to be conducted. This approach is used to determine a sequence of experiments that reduces uncertainty in the HIV infection model, and ultimately bounds the behavior of the unobserved component. Finally we provide a theoretical foundation for prediction deviation by showing that, under reasonable assumptions, it bounds the trajectory of the underlying true model.

\subsection{Confidence Intervals Do Not Measure Predictive Power}\label{sec:motivation}
Parameter confidence intervals are a poor way of determining if a nonlinear dynamical system's predictions are constrained by the observed data. 
Sensitive dependence means that tight confidence intervals do not imply constrained predictions. The classic Lorenz system provides an illustration of this phenomenon:
\begin{equation*}
\frac{dx}{dt} = \theta_1 (y-x), \quad \frac{dy}{dt} = x (\theta_2-z) - y, \quad \frac{dz}{dt} = xy - \theta_3 z.
\end{equation*}
Fig. \ref{fig:lorenz}(a) shows $x(t)$ data generated from the Lorenz system with parameters $\boldsymbol{\theta}^{\textrm{true}} = [7, 38, 5]$, initial conditions $x(0) = 10$, $y(0) = 20$, and $z(0) = 3$, and a small amount of normally distributed noise. 
The best-fit estimates for the parameters, $\boldsymbol{\theta}^*$, are very close to the true values $\boldsymbol{\theta}^{\textrm{true}}$, and have seemingly tight confidence intervals: $\theta^*_1 = 7.00$ $(6.51 - 7.49)$, $\theta^*_2 = 38.03$ $(36.08 - 40.28)$, and $\theta^*_3 = 5.00$ $(4.82 - 5.17)$, with 95\% simultaneous likelihood-based intervals in parentheses \citep{Raue09}. Suppose we wished to use these observed data with $y(0)=20$ to predict the behavior of the system when $y(0) = 7$, with all other factors staying constant. Do the tight confidence intervals allow for confidence in the model's predictions at this different initial condition? Fig. \ref{fig:lorenz}(b) shows that they do not. This figure compares the predictions made by the best-fit model $\boldsymbol{\theta}^*$ to those made by the model with parameters $\boldsymbol{\bar{\theta}} = [6.98, 38.12, 4.99]$. $\boldsymbol{\bar{\theta}}$ is well within the confidence intervals of $\boldsymbol{\theta}^*$, moreover the fit error of $\boldsymbol{\bar{\theta}}$ is within the 95\% confidence interval for the fit error of $\boldsymbol{\theta}^*$, meaning $\boldsymbol{\bar{\theta}}$ is also a good fit for the observed data. However, $\boldsymbol{\bar{\theta}}$ and $\boldsymbol{\theta}^*$ make entirely different predictions for the condition we wish to predict. The phase portraits in Figs. \ref{fig:lorenz}(c) and \ref{fig:lorenz}(d) show that this small change in the parameters is enough to send the trajectory to a different side of the attractor. The Lorenz system is a canonical example of sensitivity, but chaotic dynamics are not required to have tight confidence intervals with unconstrained predictions. For instance, this same result can be had any time a basin boundary lies within the confidence interval.

\begin{figure*}[]
\centering
\includegraphics[]{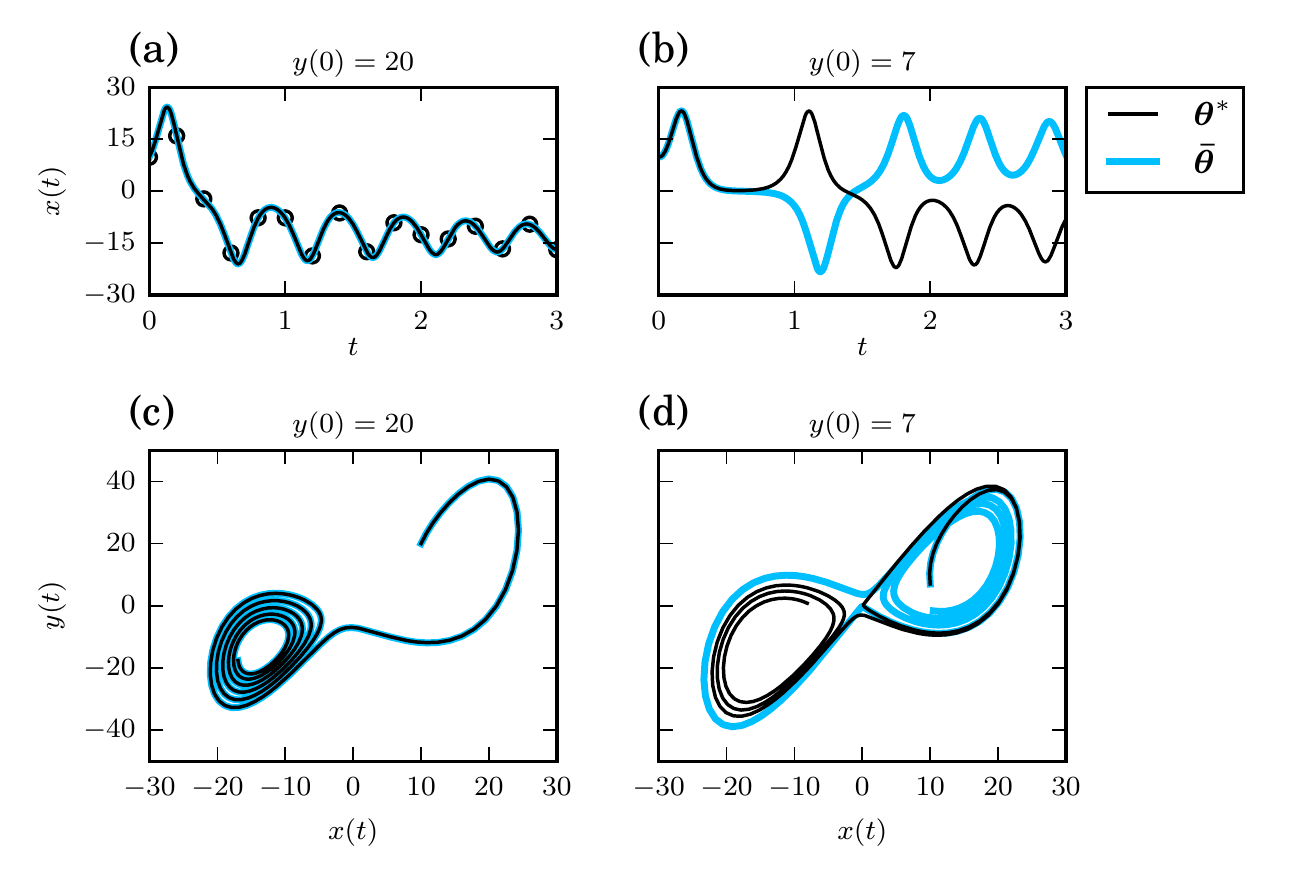}
\caption{(a) Circles indicate simulated data points from the Lorenz system, with the best-fit in black and the alternative model in blue. (b) Despite both models providing a good fit to the data in panel (a), they produce very different predictions on a different initial condition. (c, d) Phase portraits for the trajectories in panels (a) and (b).} \label{fig:lorenz}
\end{figure*}

Tight confidence intervals do not imply constrained predictions, and likewise wide confidence intervals do not imply unconstrained predictions. Parameters in nonlinear dynamical systems may be interrelated such that they individually have large confidence intervals, yet the predictions of interest are actually constrained. The following parameterization of the Lotka-Volterra predator-prey model illustrates this fact:
\begin{equation*}
\frac{dx}{dt} = \theta_1 \theta_3 x - \theta_2 \theta_3 x y,\quad \frac{dy}{dt} =  \theta_2 \theta_4 x y - \theta_1 \theta_4 y.
\end{equation*}
A symmetry in the parameters renders them all unidentifiable - they have infinite confidence intervals. Suppose we were able to observe $x(t)$ data and wished to use these data to predict the state $y(t)$. Fig. \ref{fig:lv} shows that despite the infinite confidence intervals, $x(t)$ data constrain predictions of $y(t)$. The data in Fig. \ref{fig:lv} were generated using $\boldsymbol{\theta}^{\textrm{true}} = [1, 0.05, 1, 1]$ and initial conditions $x(0) = y(0) = 10$, with standard normal noise. $\boldsymbol{\bar{\theta}}$ is the worst-case of how bad the prediction in $y(t)$ could be. Specifically, of all parameters that have fit error within the 95\% confidence interval of the fit error of the best-fit (that is, all parameters that provide a good fit of the data), $\boldsymbol{\bar{\theta}}$ is the one that maximized the squared difference between its prediction of $y(t)$ and that of the best-fit $\boldsymbol{\theta}^*$. Thus any model that fits the $x(t)$ observations will make a prediction on $y(t)$ that differs from the best-fit by no more than the difference seen in $\boldsymbol{\bar{\theta}}$.

\begin{figure*}[]
\centering
\includegraphics[]{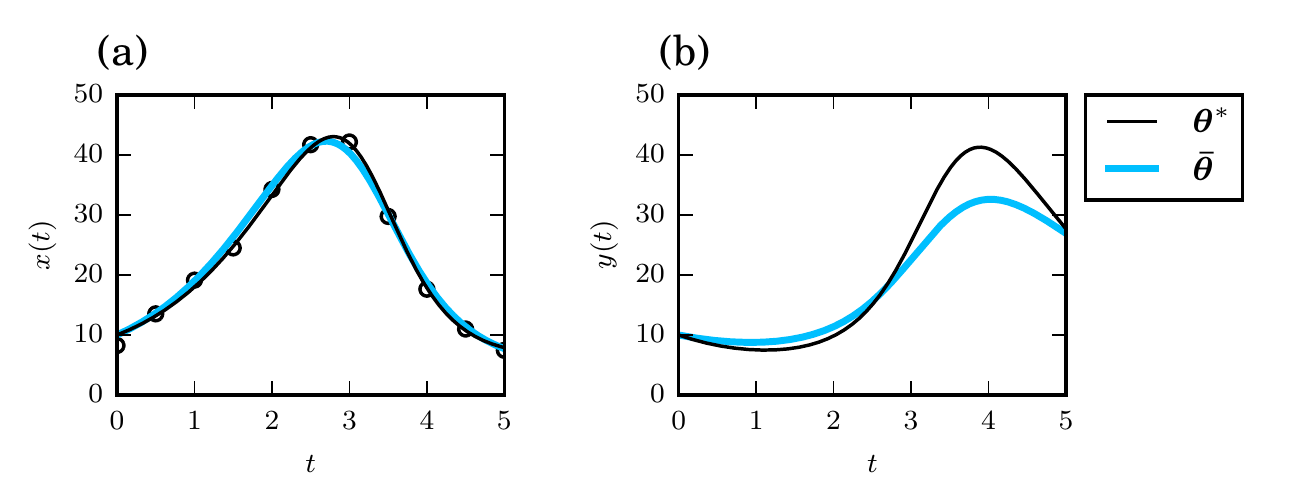}
\caption{(a) Circles indicate simulated data points from the Lotka-Volterra model, with the best-fit in black and in blue the model that maximized the difference in panel (b), subject to providing a good fit to these simulated data points. (b) Predictions from the best-fit and alternative models of the state $y(t)$ are constrained, despite unidentifiable parameters.} \label{fig:lv}
\end{figure*}

This model contains a structural unidentifiability, which could be identified and corrected by a reparameterization \citep{Bellu07, Chis11, Sedoglavic01}. It is also possible to have model parameters that are structurally identifiable but not \textit{practically} identifiable, given the noise in the collected data \citep{Raue09}. \citet{Gutenkunst07} show that models with parameters that cannot be well constrained by data are ubiquitous in systems biology, and conclude that ``modelers should focus on predictions rather than on parameters."

If the purpose of fitting the model to data is to ascertain the values of the parameters, then confidence intervals provide a useful quantification of uncertainty. However, if the purpose is to use the fitted model to make predictions about unobserved variables or unobserved conditions, then confidence intervals serve no purpose for dynamical systems. We propose putting aside the issue of measuring confidence intervals and instead directly measure the uncertainty in the quantity of interest: the predictions.

\section{Prediction Deviation as a Measure of Uncertainty}
Figs. \ref{fig:lorenz}(b) and \ref{fig:lv}(b) provide the motivation for our approach to measuring uncertainty. We consider a scenario, or set of scenarios, for which we are interested in predicting the system behavior. In Fig. \ref{fig:lorenz}(b) this was a different initial condition and in Fig. \ref{fig:lv}(b) it was an unobserved variable. We then pose the following question: Of all parameters that are a good fit to the observed data, what is the largest deviation in predicted behaviors for any pair? We call this deviation the \textit{prediction deviation}. A low prediction deviation, such as that in Fig. \ref{fig:lv}(b), means that the observed data have constrained the prediction of interest. A high prediction deviation, such as that in Fig. \ref{fig:lorenz}(b), means that the observed data have not constrained the prediction of interest.

\subsection{Parameter Estimation}
We must first introduce notation to make the idea of prediction deviation precise. We suppose that we are learning a system of ordinary differential equations with state variables $\mathbf{x}(t)$, unknown parameters $\boldsymbol{\theta}$, and known external factors $\boldsymbol{\nu}$:
\begin{equation}\label{eq:model}
\frac{d\mathbf{x}}{dt} = \mathbf{f}\left(\mathbf{x},t; \boldsymbol{\theta}, \boldsymbol{\nu} \right).
\end{equation}
If the initial conditions are known then they are included in $\boldsymbol{\nu}$, and if unknown in $\boldsymbol{\theta}$.

We now provide notation for the observed data and the data fitting problem. Let $\mathcal{P}_j = (I_j,T_j,\boldsymbol{\nu}^j)$ represent a particular \textit{experiment}, with $I_j$ the set of state variables that are observed, $T_j = \{T_{i,j} : i \in I_j\}$ the sets of time points at which these observations are made for each state variable, and $\boldsymbol{\nu}^j$ the external factors. We suppose that a total of $J$ experiments have been performed, resulting in observed data $\tilde{x}^j_i(t)$, for $j =1,\ldots,J, i \in I_j,$ and $t \in T_{i,j}$. We denote the complete set of observed experiments as $\mathcal{P} = \{\mathcal{P}_1,\ldots,\mathcal{P}_J\}$ and the complete set of observed data as $\mathbf{\tilde{x}}$.

The unknown parameters $\boldsymbol{\theta}$ are typically estimated from the observed data $\mathbf{\tilde{x}}$ by minimizing the weighted squared error
\begin{equation}\label{eq:zfit}
z_{\textrm{fit}}(\boldsymbol{\theta}; \mathcal{P}, \mathbf{\tilde{x}}) := \sum_{j=1}^J \sum_{i \in I_j} \sum_{t \in T_{i,j}} \left( \frac{x_i(t;\boldsymbol{\theta},\boldsymbol{\nu}^j) - \tilde{x}^j_i(t)}{\sigma_{ijt}} \right)^2,
\end{equation}
where $x_i(t;\boldsymbol{\theta},\boldsymbol{\nu}^j)$ is obtained by integrating (\ref{eq:model}) and $\sigma_{ijt}^2$ is the noise variance. The best-fit parameters $\boldsymbol{\theta}^*$ are the solution to least squares problem
\begin{equation}
\underset{\boldsymbol{\theta}}{\textrm{minimize}} \quad z_{\textrm{fit}}(\boldsymbol{\theta}; \mathcal{P}, \mathbf{\tilde{x}}).
\end{equation}

\subsection{Prediction Deviation}
To measure prediction deviation, we wish to search over the set of all models that provide a good fit to the observed data. We say that a model $\boldsymbol{\theta}$ is a ``good fit" to the observed data if its fit error $z_{\textrm{fit}}(\boldsymbol{\theta}; \mathcal{P}, \mathbf{\tilde{x}})$ is not too much worse than that of the best fit, $\boldsymbol{\theta}^*$. Specifically, we measure a 95\% confidence interval for $z_{\textrm{fit}}(\boldsymbol{\theta}^*; \mathcal{P}, \mathbf{\tilde{x}})$, which we denote as $[z^{*}_{l}, z^{*}_u]$. In the event of normally distributed observation noise a parametric estimate for the interval can be obtained using the $\chi^2$ distribution, or, as we do here, a nonparametric confidence interval can be obtained with the bootstrap \citep{Efron86}. The prediction deviation is defined as the maximum difference on a prediction problem between any pair of models that both have fit error within the 95\% confidence interval of the best-fit error.

The prediction problems for which the prediction deviation is to be measured are defined in the same way as the experiment according to which data were collected. We call $\mathcal{Y}_\ell = (I_\ell,T_\ell,\boldsymbol{\nu}^\ell)$ a \textit{prediction problem}, where as before $I_\ell$ is the set of state variables to be predicted in problem $\ell$, $T_\ell = \{T_{i,\ell} : i \in I_\ell\}$ are the sets of time points at which these predictions are made for each state variable, and $\boldsymbol{\nu}^\ell$ are the external factors. Let $\mathcal{Y} = \{\mathcal{Y}_1,\ldots,\mathcal{Y}_L\}$ be the full collection of variables and experiments of interest for prediction. The squared difference between $\boldsymbol{\theta}^1$ and $\boldsymbol{\theta}^2$ on the prediction problems is
\begin{align}\nonumber
z_{\textrm{dev}}(\boldsymbol{\theta}^1,& \boldsymbol{\theta}^2;\mathcal{Y}) \\
&:= \sum_{\ell=1}^L \sum_{i \in I_\ell} \sum_{t \in T_{i,\ell}} \left( \frac{x_i(t;\boldsymbol{\theta}^1,\boldsymbol{\nu}^\ell) - x_i(t;\boldsymbol{\theta}^2,\boldsymbol{\nu}^\ell)}{\sigma_{ilt}} \right)^2.
\end{align}
As before, $\sigma_{ilt}^2$ is an estimate of the noise level for that measurement, which is important primarily for combining multiple measurements of possibly different scales into one metric. Prediction deviation can now be framed as an optimization problem:
\begin{subequations}\label{prob:pred}
\begin{align}\label{eq:pred}
\underset{\boldsymbol{\theta}^1,\boldsymbol{\theta}^2}{\textrm{maximize}} \quad &z_{\textrm{dev}}(\boldsymbol{\theta}^1,\boldsymbol{\theta}^2;\mathcal{Y}) \\\label{eq:con1}
\textrm{subject to}\quad
&z_{\textrm{fit}}(\boldsymbol{\theta}^1; \mathcal{P}, \mathbf{\tilde{x}}) \leq z^{*}_u, \\\label{eq:con2}
&z_{\textrm{fit}}(\boldsymbol{\theta}^2; \mathcal{P}, \mathbf{\tilde{x}}) \leq z^{*}_u.
\end{align}
\end{subequations}
The objective (\ref{eq:pred}) searches for a pair of models that maximize the difference in predictions on the prediction problems, while the constraints (\ref{eq:con1}) and (\ref{eq:con2}) limit the search to those models that provide a good explanation for the observed data $\mathbf{\tilde{x}}$. Let  $\boldsymbol{\bar{\theta}}^1$ and $\boldsymbol{\bar{\theta}}^2$ be the maximizers of problem (\ref{prob:pred}). Then, the optimal objective value $z_{\textrm{dev}}(\boldsymbol{\bar{\theta}}^1,\boldsymbol{\bar{\theta}}^2;\mathcal{Y})$ is the prediction deviation and can be obtained by solving this single, constrained maximization problem. We show in the supplemental material \cite{Supplement} results for the Lorenz system from Fig. \ref{fig:lorenz}, and now discuss how prediction deviation can be used to understand and reduce uncertainty in a viral infection model.

\section{Uncertainty in a Model of HIV Infection}
\subsection{The Model and Data}
We now use prediction deviation to assess how well observed data constrain the model predictions of an unobserved component in a model of the innate immune response to HIV infection\citep{Browne15}. The model describes the dynamics of how interferon-alpha (IFN$\alpha$) protects CD4 T cells from infection by HIV. IFN$\alpha$ is a signaling protein that endows CD4 T cells with protection from HIV by upregulating genes that disrupt viral replication. In the model, CD4 T cells (C) are infected by HIV (H) and become infected cells (CH) that produce additional viruses. Exposure to IFN$\alpha$ (I) induces a refractory state in both uninfected and infected cells (CI and CHI respectively) which if uninfected are protected from infection, and if infected no longer produce additional viruses. The refractory state is reversible and CI and CHI cells eventually revert to their original state, C and CH respectively. The dynamical system that describes the interactions of these quantities is:
\begin{align*}
\frac{d\textrm{C}(t)}{dt} &= \theta_1 \textrm{C}(t) + \theta_3 \textrm{CI}(t) - \theta_2 \textrm{C}(t) \frac{\textrm{I}(t)}{\theta_8+\textrm{I}(t)}\\
&\quad- \theta_5 \textrm{C}(t) \textrm{H}(t),\\
\frac{d\textrm{CI}(t)}{dt} &= (\theta_1-\theta_3 ) \textrm{CI}(t)  + \theta_2 \textrm{C}(t) \frac{\textrm{I}(t)}{\theta_8+\textrm{I}(t)},\\
\frac{d\textrm{CH}(t)}{dt} &= (\theta_1- \theta_4) \textrm{CH}(t) + \theta_5 \textrm{C}(t) \textrm{H}(t)  \\
&\quad- \theta_2  \textrm{CH}(t) \frac{\textrm{I}(t)}{\theta_8+\textrm{I}(t)} + \theta_3 \textrm{CHI}(t),\\
\frac{d\textrm{CHI}(t)}{dt} &= (\theta_1-\theta_3-\theta_4) \textrm{CHI}(t) + \theta_2 \textrm{CH}(t) \frac{\textrm{I(t)}}{\theta_8+\textrm{I}(t)}, \\
\frac{d\textrm{H}(t)}{dt} &= \theta_6 \textrm{CH}(t) - \theta_7 \textrm{H}(t).
\end{align*}

We use here tissue culture data collected by \citet{Browne15}, who provide full details of the experimental methodology. In short, varying levels of IFN$\alpha$ were added to tissue cultures with CD4 T cells. After allowing the cells to incubate with the IFN$\alpha$ for six hours, HIV was added to the culture for one hour and then washed out. The total number of uninfected ($\textrm{C} + \textrm{CI}$) and infected ($\textrm{CH} + \textrm{CHI}$) cells, along with the viral count ($\textrm{H}$) were measured with four replicates every 24 hours, for 3 days. This experiment was done separately for a total of 7 initial IFN$\alpha$ levels: 0, 0.002, 0.02, 0.2, 2, 20, and 200 ng/mL. In this tissue culture the IFN$\alpha$ activity remained constant and so $\textrm{I}(t)$ was a known, external factor. Details of model fitting and prediction deviation implementation are given in Appendix \ref{Appendix:Imp}.

\subsection{Prediction Deviation after One Experiment}
To illustrate how prediction deviation changes with additional experiments and how it can be used for experiment selection, we label each combination of variables and IFN$\alpha$ level as a separate experiment. For example, $\textrm{C}+\textrm{CI}$ measured at I=0.002 ng/mL defines one experiment, and $\textrm{H}$ measured at I=2.0 ng/mL is another. There are a total of 21 such experiments for which data were collected. We begin by using data from only one of these experiments, and then consider the problems of determining uncertainty in model fit, and deciding which additional experiments to add in order to reduce prediction uncertainty.

The purpose of defining the model and collecting experimental data is to understand the dynamics of how IFN$\alpha$ provides protection to CD4 T cells during HIV infection. The experimental data themselves do not explicitly show the interaction of IFN$\alpha$ and CD4 T cells inasumch as only the sum $\textrm{C} + \textrm{CI}$ can be observed. The natural prediction problem is to then try to predict the $\textrm{CI}$ timecourse, at the same observation times as the $\textrm{C} + \textrm{CI}$ data. We begin with just one experiment, and let $\mathcal{P}$ be the experiment corresponding to $\textrm{C}+\textrm{CI}$ measured at I=0.002 ng/mL. Let $\mathcal{Y}$ be the corresponding prediction problem, $\textrm{CI}$ at I=0.002 ng/mL.

Using prediction deviation, we can determine if the observations of $\textrm{C}+\textrm{CI}$ at I=0.002 ng/mL constrain the predictions of $\textrm{CI}$ at I=0.002 ng/mL, and Fig. \ref{fig:pred_init} shows that they do not. In particular, Fig. \ref{fig:pred_init}(a) shows that the two models that maximize prediction deviation both provide a good fit for the observed data, while Fig. \ref{fig:pred_init}(b) shows that they provide widely diverging predictions about the $\textrm{CI}$ timecourse: One of the models suggests that nearly all of the CD4 T cells are refractory, while the other one suggests that nearly none of them are. These observed data do not in any way increase our understanding of the IFN$\alpha$ dynamics.

\begin{figure*}[]
\centering
\includegraphics[]{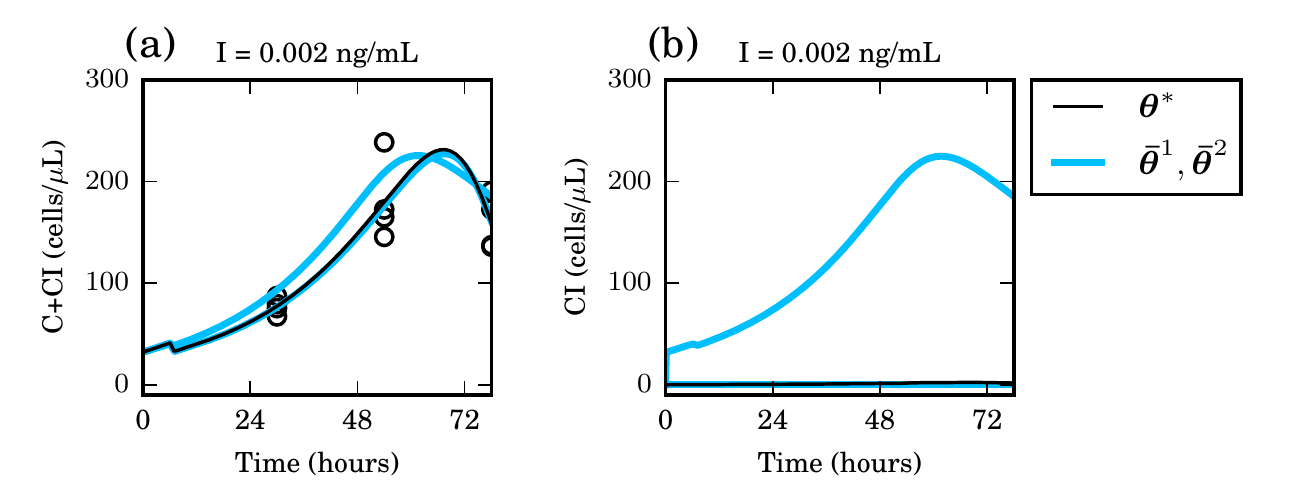}
\caption{(a) Circles indicate observed data for total uninfected CD4 T cells. In black is the best-fit model, and in blue are the two prediction deviation models, which also provide a good fit to the data. (b) The prediction deviation models provide widely differing predictions about the number of uninfected cells that are refractory, ranging from nearly none to nearly all.} \label{fig:pred_init}
\end{figure*}

\section{Optimal Experiment Design}
Knowing that the predictions of $\textrm{CI}$ are entirely unconstrained, the question that naturally follows is to determine which of the remaining 20 experiments should be done next in order to maximally reduce the prediction deviation. More generally, we wish to predict the impact that a particular candidate experiment or set of experiments $\mathcal{P}{'}$ will have on the prediction deviation, given that we have already completed experiments $\mathcal{P}$, with $\mathcal{P}{'} \cap \mathcal{P} = \emptyset$.

Fig. \ref{fig:pred_others} provides some insight into this problem. This figure shows the predictions of the best-fit and prediction deviation models from Fig. \ref{fig:pred_init} on two of the candidate experiments, $\textrm{C + CI}$ at I levels of 0.0 and 200.0 ng/mL. On the candidate experiment in Fig. \ref{fig:pred_others}(a), the prediction deviation models are very different. Suppose observations were collected for this experiment and then prediction deviation were recomputed using both these observations and the original set. After collecting data, at least one of the prediction deviation models in Fig. \ref{fig:pred_others}(a) would no longer be a good fit. We cannot know \textit{a priori} if the observations will lie close to one of the models and thereby disqualify the other, or if they will lie in the middle, disqualifying both, but at least one model will not be a good fit for the new observations.

Fig. \ref{fig:pred_others}(b) shows the alternative situation where the prediction deviation models do not disagree on the candidate experiment. Were this experiment to be done, it is possible that the observations would disqualify both prediction deviation models and there would be a reduction in uncertainty. However, it is also possible for the observations to be such that both models remain feasible, meaning there is no reduction in uncertainty.

\begin{figure*}[]
\centering
\includegraphics[]{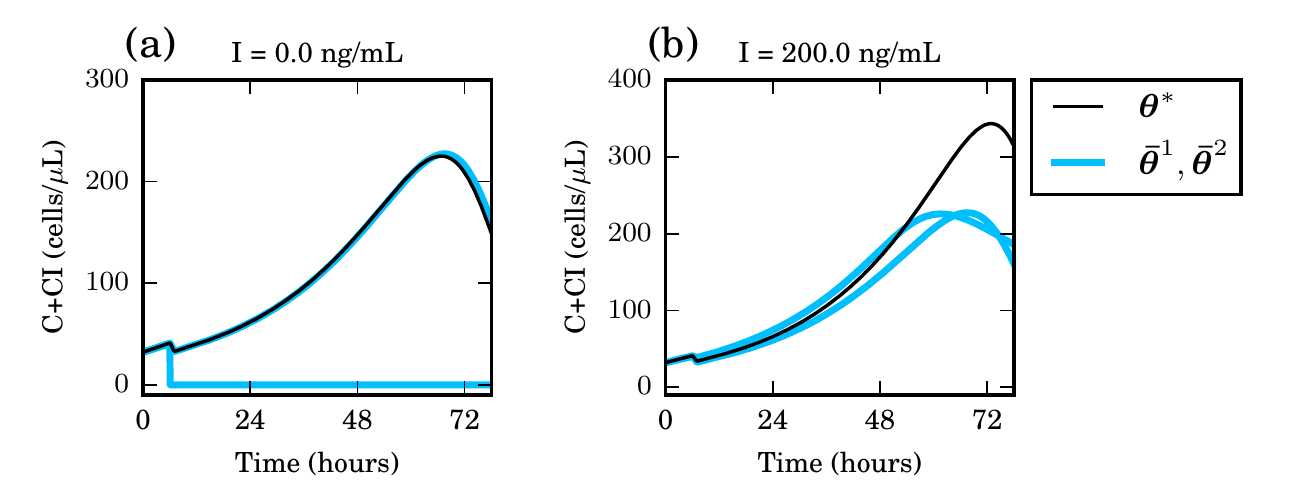}
\caption{Trajectories from the same best-fit (black) and prediction deviation (blue) models as Fig. \ref{fig:pred_init}, for two candidate experiments. (a) Observations from this experiment would disqualify at least one of the prediction deviation models. (b) Both prediction deviation models might remain feasible after this experiment.} \label{fig:pred_others}
\end{figure*}

The experiment in Fig. \ref{fig:pred_others}(a) seems like a good choice for reducing uncertainty, however having a large deviation on the candidate experiment $\mathcal{P}{'}$ does not necessarily mean that the deviation on the prediction problem of interest, in this case CI at I=0.002 ng/mL, will actually be reduced. Certainly that pair of prediction deviation models will no longer satisfy both constraints (\ref{eq:con1}) and (\ref{eq:con2}), however there may exist yet another pair of models that do not disagree on $\mathcal{P}{'}$ but produce the same prediction deviation on $\mathcal{Y}$. A powerful property of prediction deviation as a measure of uncertainty is that we actually can determine if this is the case.

\subsection{Estimating Experiment Impact}\label{sec:imp}
Collecting observations from experiment $\mathcal{P}{'}$ would change the prediction deviation by requiring the prediction deviation models to be a good fit for the new observations. In essence, there would be two new constraints that must be satisfied:
\begin{align*}
z_{\textrm{fit}}(\boldsymbol{\bar{\theta}}^1; \mathcal{P}{'}, \mathbf{\tilde{x}}{'}) &\leq \eta \quad \textrm{and}\\
z_{\textrm{fit}}(\boldsymbol{\bar{\theta}}^2; \mathcal{P}{'}, \mathbf{\tilde{x}}{'}) &\leq \eta,
\end{align*}
for some $\eta$, where $\mathbf{\tilde{x}}{'}$ are the new observations. In Appendix \ref{Appendix:Imp} we show that these constraints imply
\begin{equation}\label{eq:con3}
z_{\textrm{dev}}(\boldsymbol{\bar{\theta}}^1,\boldsymbol{\bar{\theta}}^2;\mathcal{P}{'}) \leq 4\eta,
\end{equation}
which allows us to get some idea of the impact these constraints would have even without knowledge of $\mathbf{\tilde{x}}{'}$. The essence of this result is that if the prediction deviation models are a good fit for the new data, they must have close trajectories on the new data. We are unable to restrict the prediction deviation models to be a good fit for the new data until we have collected the new data. We can, however, restrict the prediction deviation models to have close trajectories on the candidate experiment, thus estimating the impact that the candidate experiment would have on the prediction deviation. This is done by solving the prediction deviation problem with the added constraint (\ref{eq:con3}), which we call the \textit{experiment impact problem}:
\begin{subequations}\label{prob:exp}
\begin{align}
\underset{\boldsymbol{\theta}^1,\boldsymbol{\theta}^2}{\textrm{maximize}} \quad & z_{\textrm{dev}}(\boldsymbol{\theta}^1,\boldsymbol{\theta}^2;\mathcal{Y})\\
\textrm{subject to}\quad
&z_{\textrm{fit}}(\boldsymbol{\theta}^1; \mathcal{P}, \mathbf{\tilde{x}}) \leq z^{*}_u, \\
&z_{\textrm{fit}}(\boldsymbol{\theta}^2; \mathcal{P}, \mathbf{\tilde{x}}) \leq z^{*}_u, \\\label{eq:added_con}
&z_{\textrm{dev}}(\boldsymbol{\theta}^1,\boldsymbol{\theta}^2;\mathcal{P}{'}) \leq \eta.
\end{align}
\end{subequations}
This problem is identical to problem (\ref{prob:pred}) used to find the prediction deviation, with the added constraint (\ref{eq:added_con}). Model pairs like that in Fig. \ref{fig:pred_init}(a) will not be feasible solutions to this problem, inasmuch as they violate (\ref{eq:added_con}). If there does exist a different pair that produces close trajectories on $\mathcal{P}{'}$ but still has a large deviation on $\mathcal{Y}$, this optimization problem will find that pair. We denote the solutions to this optimization problem as $\boldsymbol{\hat{\theta}}^1$ and $\boldsymbol{\hat{\theta}}^2$, and call the optimal objective value $z_{\textrm{dev}}(\boldsymbol{\hat{\theta}}^1,\boldsymbol{\hat{\theta}}^2;\mathcal{Y})$ the \textit{estimated experiment impact}.

Of all possible outcomes of $\mathcal{P}{'}$, the outcome that reduces uncertainty in $\mathcal{Y}$ the least is if the observations follow the trajectories of $\boldsymbol{\hat{\theta}}^1$ and $\boldsymbol{\hat{\theta}}^2$. In this sense the predicted experiment impact is a worst-case reduction of uncertainty, and we can expect that $\mathcal{P}{'}$ will reduce the prediction deviation \textit{at least} as much as $z_{\textrm{dev}}(\boldsymbol{\hat{\theta}}^1,\boldsymbol{\hat{\theta}}^2;\mathcal{Y})$, subject to the closeness requirement $\eta$ being appropriate. Appendix \ref{Appendix:Imp} describes how $\eta$ can be chosen.

\section{Reducing Uncertainty of IFN$\alpha$ Dynamics}
We now continue the results on uncertainty in IFN$\alpha$ dynamics and use the predicted experiment impact to find additional experiments that reduce the uncertainty seen in Fig. \ref{fig:pred_init}(b). There are 20 candidate experiments consisting of different component measurements and varying IFN$\alpha$ levels. The estimated experiment impact optimization problem, (\ref{prob:exp}), was solved for each of these candidate experiments, and results for the experiment that predicted the largest reduction of uncertainty, $\textrm{C} + \textrm{CI}$ at I=0.0 ng/mL, are shown in Fig. \ref{fig:imp}. Fig. \ref{fig:imp}(a) shows the predicted experiment impact models on the candidate experiment, which are forced to have close trajectories. Fig. \ref{fig:imp}(b) shows that for the prediction problem there is a substantial reduction in uncertainty by requiring the models to produce close trajectories on the candidate experiment - this is the estimated experiment impact. The actual experiment impact is shown in Figs. \ref{fig:imp}(c) and \ref{fig:imp}(d): in Fig. \ref{fig:imp}(c) the actual observations, and in Fig. \ref{fig:imp}(d) the prediction deviation after including those observations. The actual reduction in prediction deviation was very close to that predicted by the estimated experiment impact in Fig. \ref{fig:imp}(b).

\begin{figure*}[]
\centering
\includegraphics[]{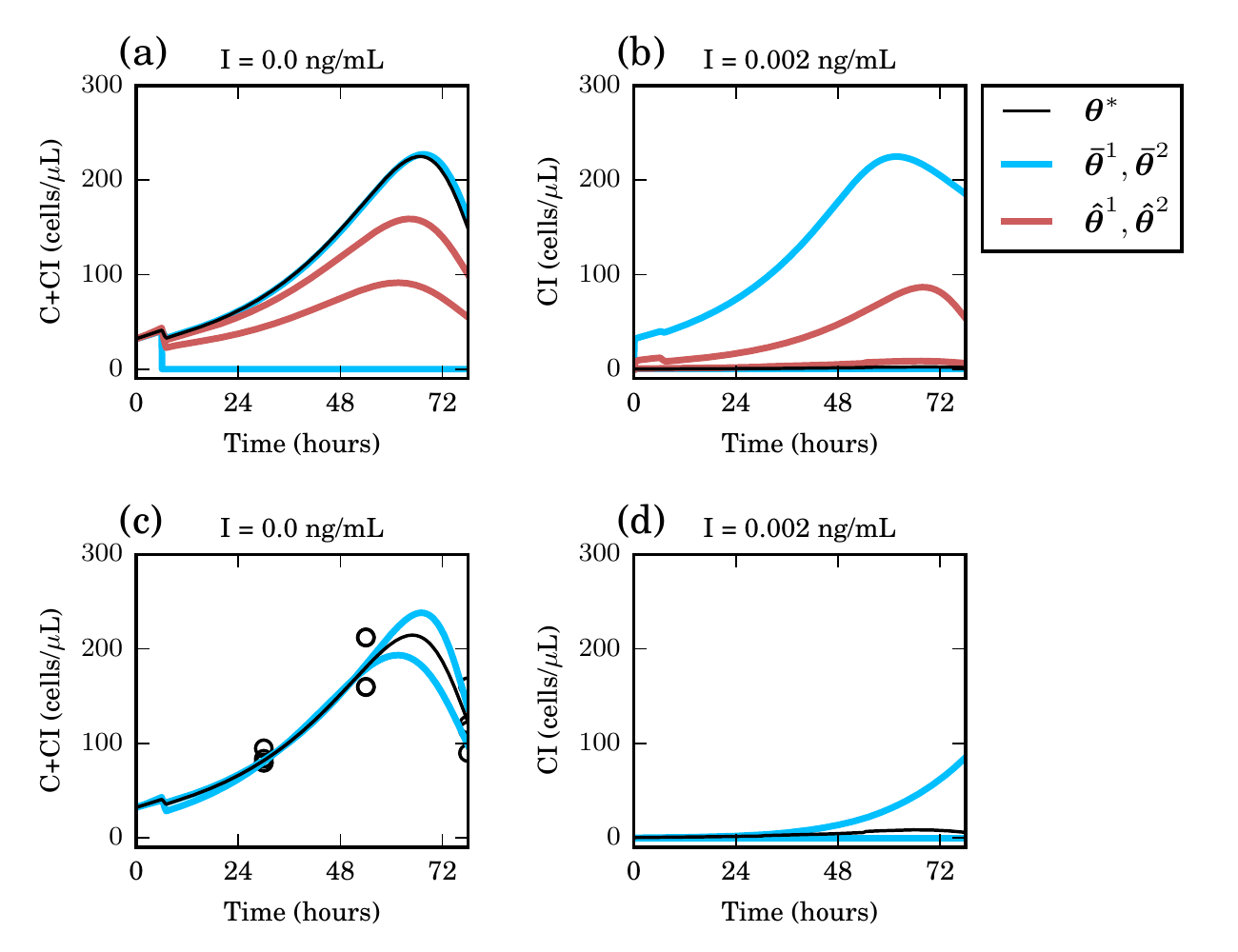}
\caption{(a, b) The expected experiment impact for a candidate experiment (a) on the prediction problem (b). In black and blue are the best-fit and prediction deviation models respectively, as in Figs. \ref{fig:pred_others}(a) and \ref{fig:pred_init}(b). In red are the expected experiment impact models, which predict a substantial reduction in uncertainty from this experiment. (c, d) The corresponding figures after adding observations from the candidate experiment in (a). In (c), the updated best-fit and prediction deviation models after adding the data (circles, with overlapping data shown side-by-side). In (d), the updated prediction deviation, reduced from (b) by the new observations.} \label{fig:imp}
\end{figure*}

The estimated experiment impact problem predicted a significant reduction in uncertainty from only one of the 20 candidate experiments. Fig. \ref{fig:expmt_imp} shows the results of separately adding each of the 20 candidates, and the $\textrm{C} + \textrm{CI}$ at I=0.0 ng/mL experiment of Fig. \ref{fig:imp} provided by far the largest reduction of uncertainty. The other two experiments at I=0.0 ng/mL provided a moderate reduction in uncertainty, while the remaining 17 experiments provided no reduction in uncertainty. Some experiments actually increased the uncertainty, by increasing the amount of noise in the fitting.

\begin{figure*}[]
\centering
\includegraphics[]{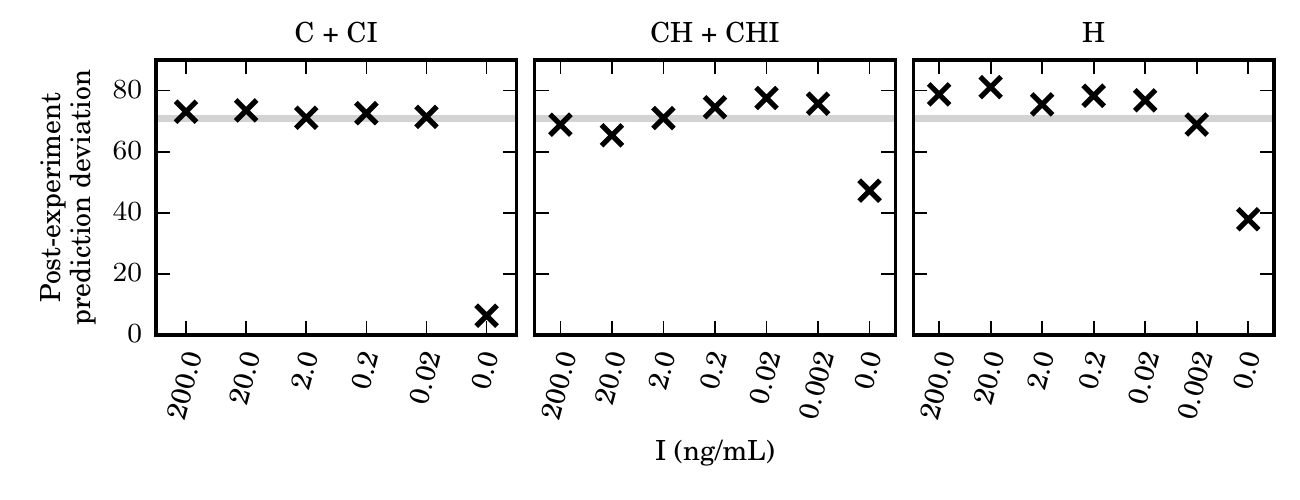}
\caption{Markers show the prediction deviation measured after including observations from each of the 20 candidate experiments. The horizontal line shows the prediction deviation prior to incorporating those observations, from Fig. \ref{fig:pred_init}. The experiment that most reduced uncertainty was that from Fig. \ref{fig:imp}.} \label{fig:expmt_imp}
\end{figure*}

We denote the prediction deviation models after including data from experiment $\mathcal{P}{'}$ as ${\boldsymbol{\bar{\theta}}^1}{'}$ and ${\boldsymbol{\bar{\theta}}^2}{'}$. Fig. \ref{fig:imp_perf}(a) compares the estimated experiment impact, $z_{\textrm{dev}}(\boldsymbol{\hat{\theta}}^1,\boldsymbol{\hat{\theta}}^2;\mathcal{Y})$, to the actual impact of each candidate experiment, $z_{\textrm{dev}}({\boldsymbol{\bar{\theta}}^1}{'},{\boldsymbol{\bar{\theta}}^2}{'};\mathcal{Y})$. As already seen in Fig. \ref{fig:imp}, for the one candidate that in actuality significantly reduced uncertainty, the predicted impact was very close to the actual impact. There were two experiments that provided a moderate reduction in uncertainty which was not matched by the estimated experiment impact. For the remaining 17 experiments, solving the estimated experiment impact problem correctly predicted that these experiments would not reduce uncertainty. Because it comes from adding a constraint to the prediction deviation problem, the estimated experiment impact problem cannot predict an increase in uncertainty, rather it can only predict that uncertainty will not decrease. Thus in Fig. \ref{fig:imp_perf}(a) the estimated experiment impacts for the 17 ineffectual candidates are very close to the previously measured prediction deviation.

The two experiments with a moderate reduction in uncertainty which was not predicted give insight into how the estimated experiment impact problem works. Fig. \ref{fig:exp7} shows the pre-experiment and post-experiment prediction deviation models, along with the expected experiment impact models, for one of these two experiments. Estimated experiment impact is a worst-case analysis, and for these two experiments the worst-case models did not reduce uncertainty while the post-experiment models did. Each of these experiments had a potential outcome, consistent with the observed data, which would not have reduced uncertainty. Fig. \ref{fig:exp7}(a) shows this worst-case potential outcome for one of the experiments. The actual data did not follow these worst-case trajectories, and in fact were able to moderately reduce uncertainty. Importantly, there were no experiments for which the estimated experiment impact indicated a reduction of uncertainty where in reality there was none. Because estimated experiment impact is a worst-case analysis, if the model is correct, this type of error will not occur and we will not do experiments that end up not reducing uncertainty.

The fact that 17 of the 20 experiments produced no reduction of uncertainty could not have been known without solving the estimated experiment impact problem. In particular, measuring the uncertainty in the candidate experiments themselves, as in Fig. \ref{fig:pred_others}, could not predict that all of these experiments would have no impact. Fig. \ref{fig:imp_perf}(b) compares the deviation on the candidate experiments, $z_{\textrm{dev}}(\boldsymbol{\bar{\theta}}^1,\boldsymbol{\bar{\theta}}^2;\mathcal{P}{'})$, to the actual experiment impact $z_{\textrm{dev}}({\boldsymbol{\bar{\theta}}^1}{'},{\boldsymbol{\bar{\theta}}^2}{'};\mathcal{Y})$. The two candidates with the highest pre-experiment deviation did not actually reduce uncertainty at all. Fig. \ref{fig:imp_perf}(b) shows that the uncertainty in the candidate experiment does not at all predict the impact that the candidate will have in the uncertainty of the prediction problem.

\begin{figure*}[]
\centering
\includegraphics[]{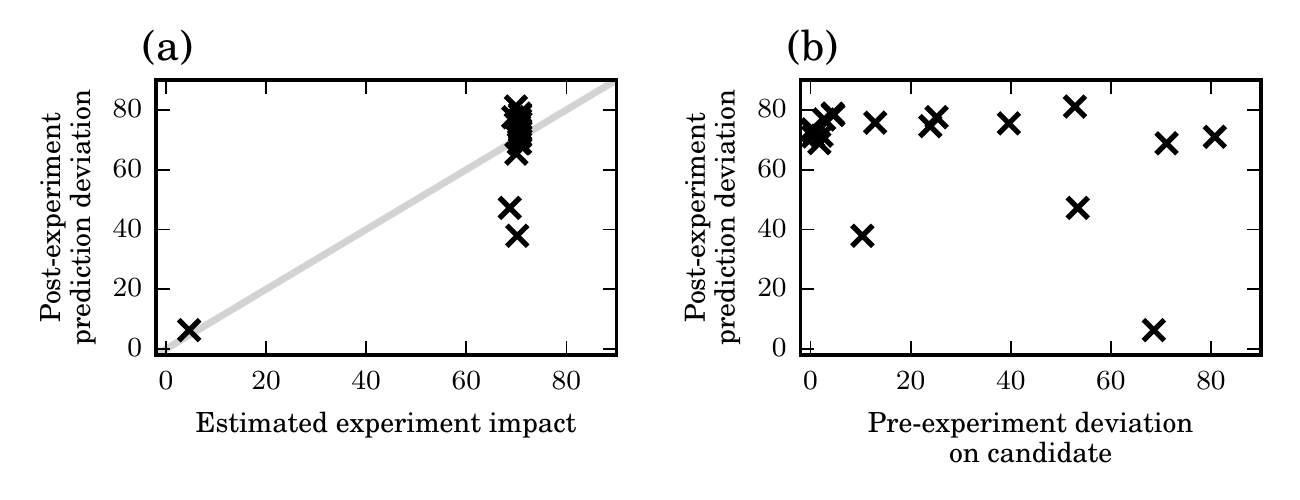}
\caption{(a) Markers show for each candidate experiment the estimated experiment impact compared to the actual prediction deviation measured after including the observations from that candidate. The gray line indicates where the estimate matches the actual outcome. (b) For each candidate experiment, the deviation of the prediction deviation models on that candidate experiment (see Fig. \ref{fig:pred_others}) compared to the prediction deviation measured after including the observations from that candidate. Candidate deviation does not provide a good prediction of experiment impact.} \label{fig:imp_perf}
\end{figure*}

\begin{figure*}[]
\centering
\includegraphics[]{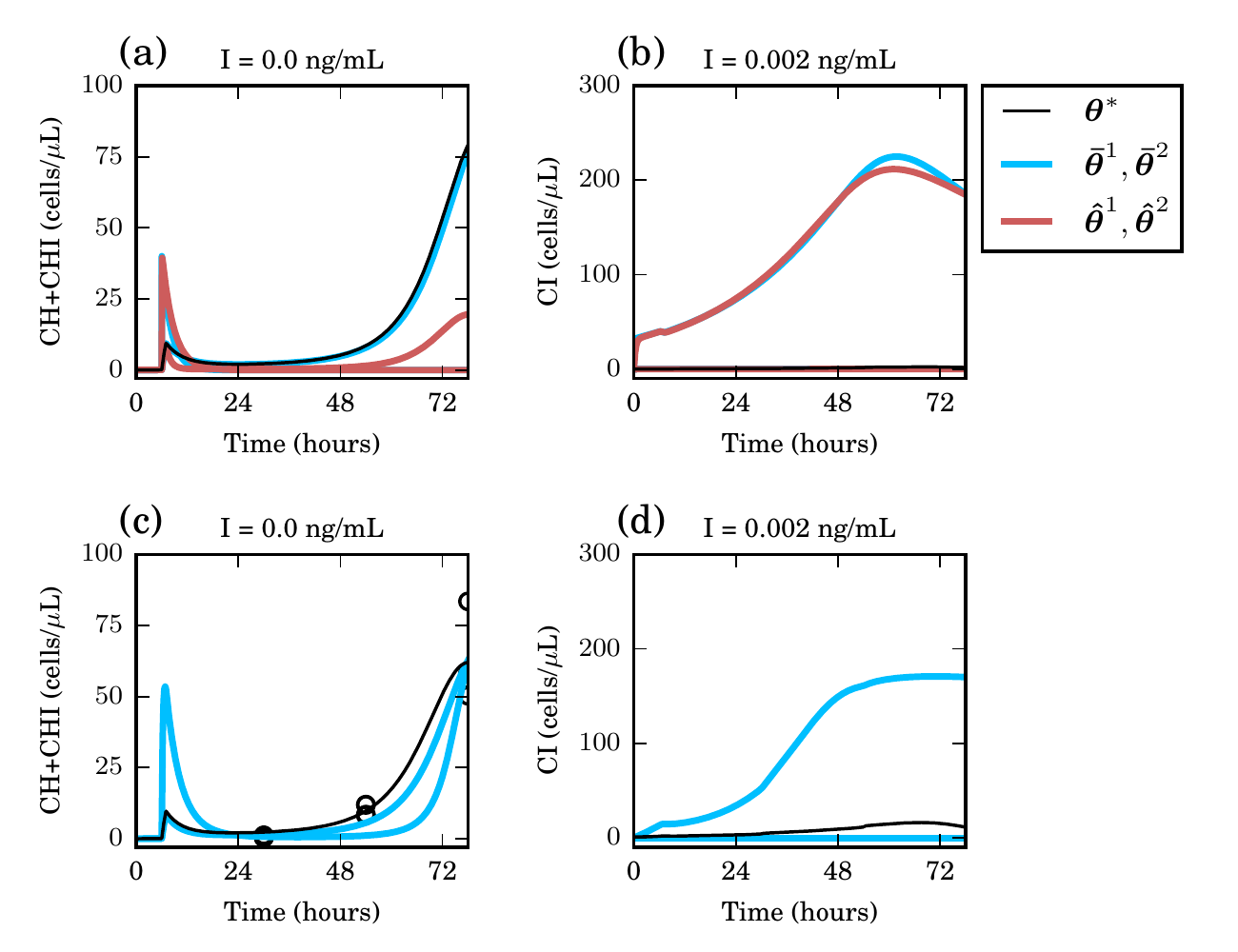}
\caption{(a, b) The expected experiment impact for a candidate experiment (a) on the prediction problem (b). In black and blue are the best-fit and prediction deviation models respectively. The expected experiment impact models (red) show a possible outcome of the experiment that does not reduce uncertainty. (c, d) The corresponding figures after adding observations from the candidate experiment in (a). The actual data from the experiment (c) were not the worst-case outcome found by the expected experiment impact problem in (a), and actually produced a moderate reduction in uncertainty (d).} \label{fig:exp7}
\end{figure*}

Fig. \ref{fig:seq_exps} shows the outcome of using the expected experiment impact in a sequential experimentation setting. Starting from the data in Fig. \ref{fig:pred_init}(a), we sequentially added in the data from the candidate experiment whose estimated experiment impact predicted the largest reduction in uncertainty. Each time after adding data from a candidate, we recomputed the prediction deviation with the new set of observations, and recomputed the estimated experiment impact of the remaining candidates. Fig. \ref{fig:seq_exps}(a) shows that adding in the second set of observations (those in Fig. \ref{fig:imp}(c)) produced a large drop in prediction deviation, seen in Figs. \ref{fig:imp}(b) and \ref{fig:imp}(d). Additional experiments continued to reduce uncertainty, but in much smaller amounts, consistent with the findings of Fig. \ref{fig:expmt_imp}. After adding data from just 3 of the 20 candidate experiments, the uncertainty was at nearly the level that was obtained by including all 20 of the candidate experiments. Fig. \ref{fig:seq_exps}(b) compares the prediction deviation with only the initial experiment to that obtained after including the first three experiments selected using the estimated experiment impact. Initially the data supported both the hypothesis that nearly none of the CD4 T cells were refractory, and the hypothesis that nearly all of the CD4 T cells were refractory. With the additional observations, the prediction deviation shows that only a small minority of CD4 T cells are refractory.

\begin{figure*}[]
\centering
\includegraphics[]{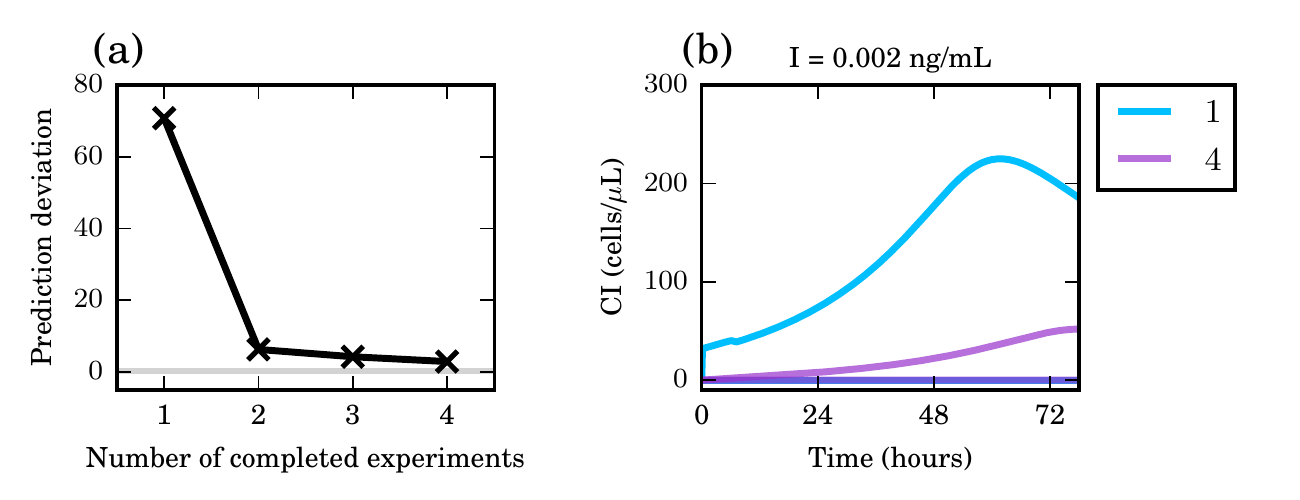}
\caption{(a) Observations were added sequentially from the candidate experiment with the best estimated experiment impact, and prediction deviation recomputed after each addition. The horizontal gray line shows the prediction deviation obtained after adding observations from all 20 candidate experiments. (b) The prediction deviation models corresponding to the 1 (blue) and 4 (purple) completed experiment markers from (a).} \label{fig:seq_exps}
\end{figure*}

\section{Theoretical Analysis}
Prediction deviation has a strong theoretical guarantee that further motivates its use as a metric of uncertainty. For the purposes of the theoretical analysis, we assume that there exists a true model $\boldsymbol{\theta}^{\textrm{true}}$, and the observed data equal the output of this true model, plus random noise:
\begin{equation*}
\tilde{x}^j_i(t) = x_i(t;\boldsymbol{\theta}^{\textrm{true}},\boldsymbol{\nu}^j) + \epsilon_{ijt},
\end{equation*}
where $\epsilon_{ijt}$ are independent but not necessarily identically distributed random variables. Let $\alpha$ be such that $z^{*}_u$ used to measure prediction deviation is the upper-bound on a $1-\alpha$ confidence interval. Under reasonable assumptions on $\epsilon_{ijt}$ which are given in Appendix \ref{Appendix:Theory}, the following theorem holds:
\begin{theorem}\label{thm}
With probability at least $1-\alpha$,
\begin{align*}
z_{\textrm{dev}}(\boldsymbol{\theta}^{\textrm{true}}, \boldsymbol{\bar{\theta}}^1;\mathcal{Y}) & \leq z_{\textrm{dev}}(\boldsymbol{\bar{\theta}}^1,\boldsymbol{\bar{\theta}}^2;\mathcal{Y}) \quad \textrm{and} \\
z_{\textrm{dev}}(\boldsymbol{\theta}^{\textrm{true}}, \boldsymbol{\bar{\theta}}^2;\mathcal{Y}) & \leq z_{\textrm{dev}}(\boldsymbol{\bar{\theta}}^1,\boldsymbol{\bar{\theta}}^2;\mathcal{Y}).
\end{align*}
\end{theorem}
This theorem means that the trajectory of the true model is in a particular sense bounded by that of the prediction deviation models: With high probability, it does not differ from either of the prediction deviation models by an amount larger than the difference in the prediction deviation models themselves. Thus if the prediction deviation is small and the trajectories of the prediction deviation models are close, then the trajectory of the true model can be specified within a narrow window, with high probability. This guarantee shows that prediction deviation corresponds to bounds on the underlying true model, and provides additional support for the validity of prediction deviation as a metric of uncertainty. The proof is given in Appendix \ref{Appendix:Theory}.

\section{Related Works}
There are several related lines of work assessing predictive power in dynamical systems. \citet{Kreutz12} use an optimization approach to measure prediction confidence intervals. Prediction intervals are measured by solving a sequence of minimization problems, separately for each time point in each prediction problem. Prediction intervals differ from the prediction deviation in that there might be different models that provide the upper and lower bounds at each time interval, whereas prediction deviation produces a single pair of models that maximizes the total deviation across all time points. The main strength of using prediction deviation as a measure of uncertainty is the ability to directly predict the impact of an additional experiment on the prediction deviation, via the estimated experiment impact problem. \citet{Kreutz12} propose using the prediction intervals of the candidate experiments to decide which experiment would have the highest impact on the prediction problem. For nonlinear dynamical systems, reducing uncertainty of the model under one condition (the candidate experiment) does not necessarily reduce uncertainty under a different condition (the prediction problem). This is shown clearly in Fig. \ref{fig:imp_perf}(b), where many candidate experiments had large uncertainty themselves, yet their observations did not reduce the uncertainty in the prediction problem. For prediction deviation, on the other hand, solving the optimization problem in (\ref{prob:exp}) provides a direct estimate of how much reducing uncertainty in the proposed experiment will reduce uncertainty in the prediction problem. Because it is a worst-case analysis, the estimate from solving (\ref{prob:exp}) also will not make the sort of error shown in Fig. \ref{fig:imp_perf}(b) where the recommended experiments end up not reducing uncertainty. Other approaches to measuring prediction intervals include boostrapping \citep{StJohn13} and MCMC sampling in a Bayesian framework \citep{Vanlier12}. \citet{Vanlier13} provide a review of recent approaches to measuring uncertainty both in parameters and in predictions. 

Optimal experimental design has typically been studied in the context of parameter estimation \citep{Bandara09, Raue10, Transtrum12} or, more recently, model selection and discrimination \citep{Flassig12, Vanlier14, Busetto13, Daunizeau11, Skanda10}. \citet{Kreutz09} provide a review of recent approaches to optimal experimental design for these two problems. Our methods here are for optimal experimental design for prediction uncertainty, which generally requires predicting the impact of a proposed experiment on prediction uncertainty. \citet{Casey07} measure prediction uncertainty using a linearization of the prediction problem, and then show how to predict the impact of a proposed experiment on the approximated prediction uncertainty. Another approach to optimal experiment design is to simulate the outcome of the proposed experiment using the best-fit model, and to measure the corresponding reduction in uncertainty \citep{Transtrum12}. Useful experiments will themselves have high prediction uncertainty, so there will likely be a large range of possible outcomes, only one of which is the best-fit outcome. As seen in Fig. \ref{fig:exp7}, the impact of the experiment on prediction uncertainty may depend strongly on which of the possible outcomes is realized. The actual reduction of uncertainty from an experiment could be much less than that predicted by the best-fit outcome, potentially wasting a valuable experiment. Solving (\ref{prob:exp}) measures uncertainty under the worst-case of the possible outcomes of the experiment, ensuring that the experiment will be useful whatever the outcome may be.

\section{Conclusions}
Two important questions that arise when fitting nonlinear dynamical systems to data are uncertainty quantification and optimal experimental design. We presented in this paper a prediction-centered approach for measuring uncertainty in a dynamical system's fit to data. Prediction deviation is able to directly show, via the pair of prediction deviation models, how much uncertainty remains in the prediction problem, thus answering the uncertainty quantification question. Solving the estimated experiment impact problem provides \textit{a priori} a direct estimate of the impact that a candidate experiment would have on uncertainty. This allows the experimenter to choose the additional experiments that are likely to most reduce uncertainty, thus answering the question of optimal experimental design. We used the estimated experiment impact problem to sequentially choose 4 experiments which produced nearly the same reduction in uncertainty as the full set of 20 candidate experiments. In addition to the sequential experimentation setting that was demonstrated here, estimated experiment impact can also be used to predict the impact of simultaneously running a number of experiments by combining them into a single candidate. Finally, we proved a bound that with high probability provides a direct relationship between prediction deviation and how constrained the underlying true model is, providing a theoretical foundation for using prediction deviation as a metric of uncertainty.

\begin{acknowledgments}
Funding for this project was provided in part by Siemens and the US Army Research Office.
\end{acknowledgments}

\appendix

\section{Implementation Details}\label{Appendix:Imp}
\subsection{Specifying the Parameter $\eta$}
Observations $\mathbf{\tilde{x}}'$ from candidate experiment $\mathcal{P}'$ would constrain the prediction deviation models according to the two constraints $z_{\textrm{fit}}(\boldsymbol{\bar{\theta}}^1; \mathcal{P}', \mathbf{\tilde{x}}') \leq \eta$ and $z_{\textrm{fit}}(\boldsymbol{\bar{\theta}}^2; \mathcal{P}', \mathbf{\tilde{x}}') \leq \eta$. The amount that the fit error on the new models would be constrained, $\eta$, is a parameter in the estimated experiment impact problem, (\ref{prob:exp}). Assuming normally distributed noise and a reasonable estimate of the experiment noise level $\sigma_{ijt}^2$, $z_{\textrm{fit}}(\boldsymbol{\theta}^*; \mathcal{P}', \mathbf{\tilde{x}}')$ follows a $\chi^2$ distribution whose $95\%$ percentile provides a reasonable choice for $\eta$. Alternatively, since observations are normalized by their noise level when computing fit error, if all observations contribute equally to the uncertainty then $\eta = z^*_u |\mathcal{P}'|/|\mathcal{P}|$ provides a reasonable choice, where $|\mathcal{P}|$ is the number of observations in experiment $\mathcal{P}$ and $z^*_u$ is the upper end of the $95\%$ confidence interval for the best-fit error. This is the approach we used in our experiments, and the effect of this $\eta$ through (\ref{eq:added_con}) can be seen in Fig. \ref{fig:imp}(a).

The following result provides the motivation for constraint (\ref{eq:added_con}) in the estimated experiment impact problem.
\begin{proposition}
$z_{\textrm{fit}}(\boldsymbol{\bar{\theta}}^1; \mathcal{P}', \mathbf{\tilde{x}}') \leq \eta$ and $z_{\textrm{fit}}(\boldsymbol{\bar{\theta}}^2; \mathcal{P}', \mathbf{\tilde{x}}') \leq \eta$ imply $z_{\textrm{dev}}(\boldsymbol{\bar{\theta}}^1,\boldsymbol{\bar{\theta}}^2;\mathcal{P}') \leq 4\eta$.
\end{proposition}
\begin{proof}
\begin{align*}
z_{\textrm{dev}}(\boldsymbol{\bar{\theta}}^1,& \boldsymbol{\bar{\theta}}^2;\mathcal{P}') \\
&= \sum_{j=1}^J \sum_{i \in I_j} \sum_{t \in T_{i,j}}  \left( \frac{x_i(t;\boldsymbol{\bar{\theta}}^1,\boldsymbol{\nu}^\ell) - x_i(t;\boldsymbol{\bar{\theta}}^2,\boldsymbol{\nu}^\ell)}{\sigma_{ilt}} \right)^2 \\
&= \sum_{j=1}^J \sum_{i \in I_j} \sum_{t \in T_{i,j}}  \left( \frac{(x_i(t;\boldsymbol{\bar{\theta}}^1,\boldsymbol{\nu}^\ell) - \tilde{x}^j_i(t))}{\sigma_{ilt}} \right. \\
&\quad\quad\quad\quad\quad\quad\quad\quad - \left. \frac{(x_i(t;\boldsymbol{\bar{\theta}}^2,\boldsymbol{\nu}^\ell)- \tilde{x}^j_i(t) )}{\sigma_{ilt}} \right)^2 \\
& \leq \left( \sqrt{z_{\textrm{fit}}(\boldsymbol{\bar{\theta}}^1; \mathcal{P}', \mathbf{\tilde{x}}')} + \sqrt{z_{\textrm{fit}}(\boldsymbol{\bar{\theta}}^2; \mathcal{P}', \mathbf{\tilde{x}}')} \right)^2 \\
& \leq 4\eta.
\end{align*}
The third line uses the triangle inequality and the last line is by supposition.
\end{proof}
This result allows for an approximation of the impact that $\mathcal{P}'$ that does not require knowledge of the data $\mathbf{\tilde{x}}'$. The triangle inequality is generally loose, and incorporating this constraint into a maximization problem means that the result is the worst-case impact of $\mathbf{\tilde{x}}'$. These two approximations provide room for the additional approximation made above in choosing $\eta$. Ultimately Fig. \ref{fig:imp_perf}(a) shows that the approximations involved in estimating experiment impact are good enough to be useful.

\subsection{Simulation and Optimization}
SloppyCell \citep{SloppyCell1, SloppyCell2} was used to integrate the model ODE system. In addition to integrating the model, SloppyCell integrates the forward sensitivity system, which provides gradients of the model trajectories with respect to the parameters, $\nabla_{\boldsymbol{\theta}} x_i(t; \boldsymbol{\theta}, \boldsymbol{\nu}^j)$. From these gradients, it is a straightforward calculation to obtain the gradients of the objectives and constraints for the three optimization problems in this paper: the data fitting problem, the prediction deviation problem, and the estimated experiment impact problem. All optimization problems were solved using random restarts of gradient-based optimization methods, with each optimization problem solved from 20 random initializations. The data fitting problem is an unconstrained minimization problem, and was solved using the Scipy implementation of the Newton conjugate-gradient algorithm \citep{Nash84, Nocedal06}. The prediction deviation and estimated experiment impact problems are constrained maximization problems and were solved using the logarithmic barrier method \citep[Framework 17.2]{Nocedal06}. This method solves the constrained problem via a sequence of unconstrained problems, each of which was solved using the Newton conjugate-gradient method. The computational difficulty of each of these unconstrained problems is similar to that of the data-fitting problem. Solving (\ref{prob:pred}) and (\ref{prob:exp}) should thus scale in a similar way as the data fitting problem, and have similar challenges. Feasible initial values for the prediction deviation and estimated experiment impact optimization problems were obtained using a Gaussian random walk from the best-fit parameters (which are always feasible), rejecting infeasible steps.

\subsection{Experimental Data}
The data for the experiment on IFN$\alpha$ dynamics were those provided by \citet{Browne15}. There two parameters were measured separately from these data, and we followed and treated these parameters, as well as all initial conditions, as known. One of the known parameters was the IFN$\alpha$ decay rate, and so $\textrm{I}(t)$ was thus known. The estimation done in this paper was then on a space of 7 parameters and 5 variables. The noise variance estimate used for weighted least squares and for prediction deviation, $\sigma_{ijt}^2$, was taken as the average over time of the sample variances across the four replicates at each time point, separately for each set of variables and IFN$\alpha$ level. This is equivalent to the maximum likelihood estimate under a model where the noise is normally distributed with a variance that differs across variables and IFN$\alpha$ levels but is constant across time points.

\section{Proof of the Theoretical Result}\label{Appendix:Theory}
The result of Theorem \ref{thm} provides a theoretical foundation for using prediction deviation as a metric of uncertainty by showing that it relates directly to bounds on the behavior of the underlying true model. The theorem requires the following assumptions:
\begin{assumption}
The observed data are the output of a true model $\boldsymbol{\theta}^{\textrm{true}}$, plus noise: $\tilde{x}^j_i(t) = x_i(t;\boldsymbol{\theta}^{\textrm{true}},\boldsymbol{\nu}^j) + \epsilon_{ijt}$.
\end{assumption}
\begin{assumption}\label{ass:ind}
The random variables $\epsilon_{ijt}$ are independent.
\end{assumption}
\begin{assumption}\label{ass:distr}
The probability density function of $\epsilon_{ijt}$ is symmetric about $0$ and unimodal, meaning the distribution function $F_{\epsilon_{ijt}}(x)$ is convex for $x\leq 0$ and concave for $x \geq 0$.
\end{assumption}
\begin{assumption}\label{ass:ci}
Let $\boldsymbol{\theta}^*$ be the best-fit model under a particular realization of the observations and $z_u^*$ fixed. Then, assume $\mathbb{P}_{\mathbf{\tilde{x}}}(z_{\textrm{fit}}(\boldsymbol{\theta}^*; \mathcal{P}, \mathbf{\tilde{x}}) \leq z_u^*) \geq 1-\alpha$. 
\end{assumption}
Assumption \ref{ass:ind} requires independence, but does not require $\epsilon_{ijt}$ to be identically distributed, thus the noise level may vary across different observations. Assumption \ref{ass:distr} is quite general: it is satisfied by the normal distribution, as well as by other heavy-tailed distributions. In Assumption \ref{ass:ci}, the model $\boldsymbol{\theta}^*$ is held constant and the randomness is over different realizations of $\epsilon_{ijt}$, and thus different realizations of $\mathbf{\tilde{x}}$. This assumption is about how the best-fit to one realization of the data generalizes to other realizations of the data, and requires that $z_u^*$, used in constraints (\ref{eq:con1}) and (\ref{eq:con2}), actually provides a $1-\alpha$ upper bound for the fit error.

For the proof of Theorem \ref{thm}, we define notation to describe the squared residuals. Let $R^{\textrm{true}}_{ijt} = \left(x_i(t;\boldsymbol{\theta}^{\textrm{true}},\boldsymbol{\nu}^j) - \tilde{x}^j_i(t)\right)^2$ be the squared residuals under the true model and $R^*_{ijt}$ the squared residuals under the best-fit model, $\boldsymbol{\theta}^*$. Let $b_{ijt} = x_i(t;\boldsymbol{\theta}^*,\boldsymbol{\nu}^j) - x_i(t;\boldsymbol{\theta}^{\textrm{true}},\boldsymbol{\nu}^j)$ be the bias of the best-fit model.

The following result shows that intervals of the noise distribution centered on $0$ contain the most probability mass.
\begin{lemma}\label{lemma:1}
For $x \geq 0$ and for $a \in \mathbb{R}$, $F_{\epsilon_{ijt}}(x + a) - F_{\epsilon_{ijt}}(-x + a) \leq F_{\epsilon_{ijt}}(x) - F_{\epsilon_{ijt}}(-x)$.
\end{lemma}
\begin{proof}
This result follows from Assumption \ref{ass:distr}. When $x=0$ the result is trivial. For $x>0$, we first consider the case where $a \geq x$. For all $x\geq 0$, $F_{\epsilon_{ijt}}(x)$ is concave, and thus $F_{\epsilon_{ijt}}'(x)$ is monotonically non-increasing. This means $\frac{\partial}{\partial a} \left(F_{\epsilon_{ijt}}(x+a) - F_{\epsilon_{ijt}}(-x+a) \right) \leq 0$ $\forall a \geq x$, and this quantity is maximized when $a=x$. Thus,
\begin{align*}
F_{\epsilon_{ijt}}(x+a) - &F_{\epsilon_{ijt}}(-x+a) \\
&\leq F_{\epsilon_{ijt}}(2x) - F_{\epsilon_{ijt}}(0) \\
& \leq 2(F_{\epsilon_{ijt}}(x) - F_{\epsilon_{ijt}}(0))\\
& = F_{\epsilon_{ijt}}(x) + 1 - F_{\epsilon_{ijt}}(-x) - 2F_{\epsilon_{ijt}}(0)\\
& = F_{\epsilon_{ijt}}(x) - F_{\epsilon_{ijt}}(-x),
\end{align*}
which is the statement of the lemma. The second line follows directly from the concavity of $F_{\epsilon_{ijt}}(x)$ and the third line uses the symmetry $F_{\epsilon_{ijt}}(x) = 1 - F_{\epsilon_{ijt}}(-x)$. When $a \leq -x$, the same argument holds using the convexity of $F_{\epsilon_{ijt}}(x)$ for $x \leq 0$.

For the remaining case, $|a| < x$, 
\begin{equation*}
F_{\epsilon_{ijt}}(x) \geq \frac{1}{2} \left( F_{\epsilon_{ijt}}(x+a) + F_{\epsilon_{ijt}}(x-a) \right)
\end{equation*}
by the concavity of $F_{\epsilon_{ijt}}(x)$ on the interval $[x-a, x+a]$. From the symmetry, it then follows that
\begin{align*}
1 + F_{\epsilon_{ijt}}(x) - F_{\epsilon_{ijt}}(-x) &\geq F_{\epsilon_{ijt}}(x+a) + 1\\
&\quad - F_{\epsilon_{ijt}}(-x+a).
\end{align*}
After rearranging, this proves the lemma.
\end{proof}
An important concept for the proof of Theorem \ref{thm} is that of a stochastic ordering, which we now define and then use to prove the theorem.
\begin{definition}
For random variables $X$ and $Y$, $X \preceq Y$ if $\mathbb{P}(X > x) \leq \mathbb{P}(Y > x) \enskip \forall x$.
\end{definition}
\begin{lemma}\label{lem:order}
$R^{\textrm{true}}_{ijt} \preceq R^*_{ijt}$.
\end{lemma}
\begin{proof}
\begin{align*}
\mathbb{P}(R^*_{ijt} \leq x) &= \mathbb{P}\left((\epsilon_{ijt} - b_{ijt})^2 \leq x \right) \\
& = F_{\epsilon_{ijt}}(\sqrt{x} + b_{ijt}) - F_{\epsilon_{ijt}}(-\sqrt{x} + b_{ijt})\\
& \leq F_{\epsilon_{ijt}}(\sqrt{x}) - F_{\epsilon_{ijt}}(-\sqrt{x}) \\
& = \mathbb{P}(R^{\textrm{true}}_{ijt} \leq x),
\end{align*}
using Lemma \ref{lemma:1}.
\end{proof}
The next result comes from Shaked and Shanthikumar \cite{Shaked07}, Theorem 1.A.3(b).
\begin{lemma}\label{lem:sum}
For independent random variables $X_1, \ldots, X_n$ and $Y_1, \ldots, Y_n$, let $X = \sum_{i=1}^n w_i X_i$ and $Y = \sum_{i=1}^n w_i Y_i$ with non-negative weights $w_1, \ldots, w_n$. If $X_i \preceq Y_i \enskip \forall i$, then $X \preceq Y$.
\end{lemma}
\begin{corollary}\label{cor:order}
$z_{\textrm{fit}}(\boldsymbol{\theta}^{\textrm{true}}; \mathcal{P}, \mathbf{\tilde{x}}) \preceq z_{\textrm{fit}}(\boldsymbol{\theta}^*; \mathcal{P}, \mathbf{\tilde{x}})$.
\end{corollary}
\begin{proof}
The fit error is a weighted sum of the squared residuals, with weights $\frac{1}{\sigma_{ijt}^2}$, so this result follows directly from Lemmas \ref{lem:order} and \ref{lem:sum}, and Assumption \ref{ass:ind}.
\end{proof}
\begin{repeattheorem}
With probability at least $1-\alpha$,
\begin{align*}
z_{\textrm{dev}}(\boldsymbol{\theta}^{\textrm{true}}, \boldsymbol{\bar{\theta}}^1;\mathcal{Y}) & \leq z_{\textrm{dev}}(\boldsymbol{\bar{\theta}}^1,\boldsymbol{\bar{\theta}}^2;\mathcal{Y}) \quad \textrm{and} \\
z_{\textrm{dev}}(\boldsymbol{\theta}^{\textrm{true}}, \boldsymbol{\bar{\theta}}^2;\mathcal{Y}) & \leq z_{\textrm{dev}}(\boldsymbol{\bar{\theta}}^1,\boldsymbol{\bar{\theta}}^2;\mathcal{Y}).
\end{align*}
\end{repeattheorem}
\begin{proof}
By Corollary \ref{cor:order} and Assumption \ref{ass:ci},
\begin{align*}
\mathbb{P}(z_{\textrm{fit}}(\boldsymbol{\theta}^{\textrm{true}}; \mathcal{P}, \mathbf{\tilde{x}}) \leq z_u^*) &\geq \mathbb{P}(z_{\textrm{fit}}(\boldsymbol{\theta}^*; \mathcal{P}, \mathbf{\tilde{x}}) \leq z_u^*) \\
& \geq 1-\alpha.
\end{align*}
Thus with probability at least $1-\alpha$, $(\boldsymbol{\theta}^{\textrm{true}}, \boldsymbol{\bar{\theta}}^1)$ is a feasible solution to problem (\ref{prob:pred}). The proof of the theorem is then by contradiction: If $z_{\textrm{dev}}(\boldsymbol{\theta}^{\textrm{true}}, \boldsymbol{\bar{\theta}}^1;\mathcal{Y}) > z_{\textrm{dev}}(\boldsymbol{\bar{\theta}}^1,\boldsymbol{\bar{\theta}}^2;\mathcal{Y})$, then $(\boldsymbol{\bar{\theta}}^1,\boldsymbol{\bar{\theta}}^2)$ cannot be an optimal solution to problem (\ref{prob:pred}). However, $(\boldsymbol{\bar{\theta}}^1,\boldsymbol{\bar{\theta}}^2)$ are defined to be optimal solutions, and so the theorem holds. The same argument simultaneously holds for $(\boldsymbol{\theta}^{\textrm{true}}, \boldsymbol{\bar{\theta}}^2)$.
\end{proof}


%

\newpage\phantom{blabla}\newpage\phantom{blabla}

\section{Supplementary material}
We show in the supplemental material prediction deviation results for the Lorenz system from Fig. 1.

\renewcommand{\thefigure}{S\arabic{figure}}

\begin{figure}
\centering
\includegraphics[]{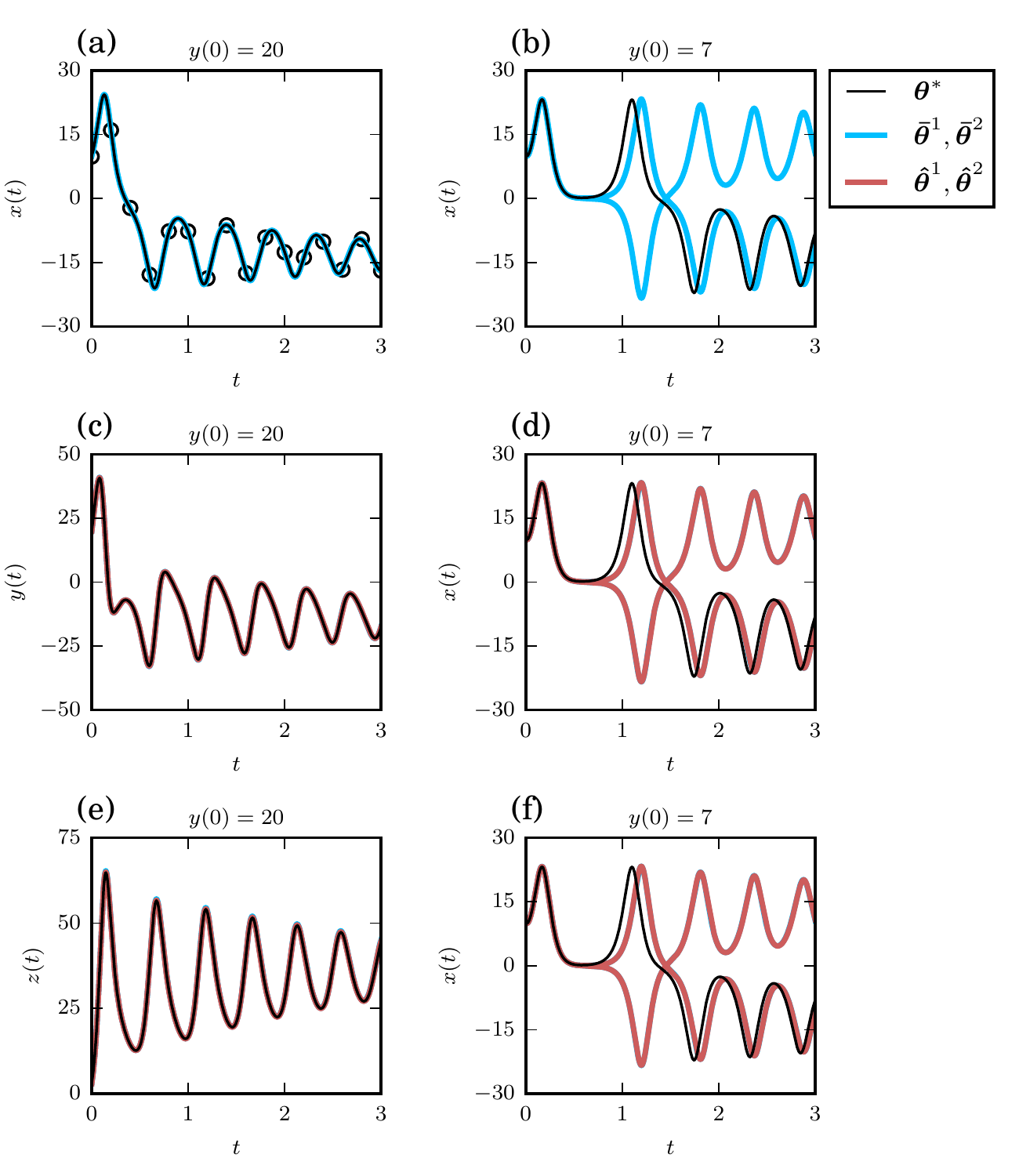}
\caption{Prediction deviation and estimated experiment impact in the Lorenz system. (a) Circles show observed data $x(t)$ with the initial condition $y(0) = 20$. The black line shows the best fit, and blue lines show the prediction deviation models for the prediction problem in (b), which is to predict $x(t)$ when $y(0) = 7$. Prediction deviation shows that the observations in (a) do not constrain the prediction problem. Panels (c) and (e) show the estimated experiment impact models (red) for candidate experiments in which (respectively) $y(t)$ and $z(t)$ are measured, with $y(0) = 20$. In (d) and (f) are the results on the prediction task. Additional measurements of $y(t)$ and $z(t)$ at $y(0) = 20$ can still leave the predictions of $x(t)$ at $y(0) = 7$ unconstrained.} \label{S1}
\end{figure}

\begin{figure}
\centering
\includegraphics[]{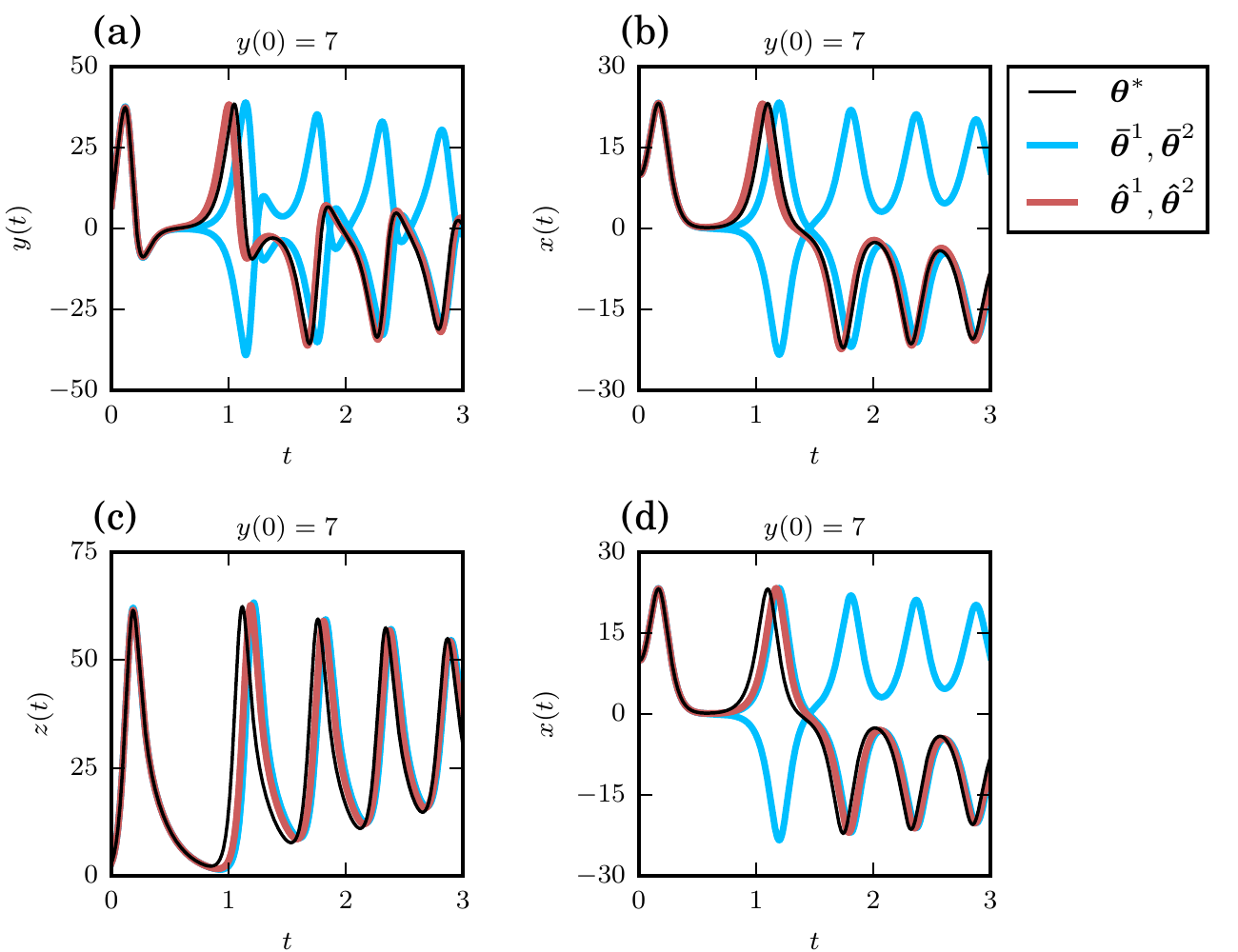}
\caption{Additional candidate experiments for the observations in Fig. \ref{S1}(a) and prediction problem in Fig. \ref{S1}(b). Observing either $y(t)$ at $y(0) = 7$ (a) or $z(t)$ at $y(0) = 7$ (c) is sufficient to constrain predictions of $x(t)$ at $y(0) = 7$. Panels (b) and (d) show the estimated experiment impact models (red) corresponding to (a) and (c), respectively.}
\end{figure}

\end{document}